\documentclass[journal]{IEEEtran}

\usepackage{amsmath,graphicx,color,psfrag,mathrsfs,bm}
\usepackage{hyperref}
\usepackage{floatrow}%
\usepackage{longtable}          
\usepackage{amsthm}
\usepackage{amssymb}
\usepackage{psfrag}
\usepackage{pstool}
\usepackage{multicol}

\newtheorem{thm}{Theorem}[section]

\newtheorem{lem}[thm]{Lemma}
\newtheorem{prop}[thm]{Proposition}
\newtheorem{assumption}[thm]{Assumption}
\theoremstyle{definition}

\theoremstyle{remark}
\newtheorem{rem}{Remark}[section]

\theoremstyle{definition}
\newtheorem{exmp}{Example}[section]
\theoremstyle{definition}
\newtheorem{prb}{Problem}[section]

\begin{document}
%
\title{Theoretical Bounds in Minimax Decentralized Hypothesis Testing}

\author{G\"{o}khan G\"{u}l,~\IEEEmembership{Student~Member,~IEEE,}
        Abdelhak M. Zoubir,~\IEEEmembership{Fellow,~IEEE}
\thanks{G. G\"{u}l and A. M. Zoubir are with the Signal Processing Group, Institute
of Telecommunications, Technische Universit\"{a}t Darmstadt, 64283, Darmstadt,
Germany (e-mail: ggul@spg.tu-darmstadt.de; zoubir@spg.tu-darmstadt.de)}
\thanks{Manuscript received April 21, 2016; revised Jane 21, 2016.}}

\maketitle

\begin{abstract}
Minimax decentralized detection is studied under two scenarios: with and without a fusion center when the source of uncertainty is the Bayesian prior. When there is no fusion center, the constraints in the network design are determined. Both for a single decision maker and multiple decision makers, the maximum loss in detection performance due to minimax decision making is obtained. In the presence of a fusion center, the maximum loss of detection performance between with- and without fusion center networks is derived assuming that both networks are minimax robust. The results are finally generalized.

\end{abstract}

\begin{IEEEkeywords}
Robustness, Distributed Detection, Data Fusion, Consensus, Sensor Networks
\end{IEEEkeywords}

%
\IEEEpeerreviewmaketitle


\section{Introduction}
\label{sec:intro}
Detecting events of interest, which is often carried out by hypothesis testing, is fundamental in many areas such as radar, sonar and digital communications. With a Bayesian framework for detection, testing a simple null hypothesis against a simple alternative requires the assumption that both the a priori probabilities (priors) and the exact probability distributions under each hypothesis are known. However, in practice, such assumptions often do not hold as there are deviations from the model assumptions, such as occurs with the presence of outliers or inaccurate a priori information. In order to deal with such situations, an optimum decision rule should take into account the imprecise knowledge of the nominal distributions as well as the a priori probabilities of the hypotheses \cite{Levy}.\\
A statistical test, which maintains a certain level of detection performance, regardless of the uncertainties in the assumed model, is known as minimax detection. While the term \it minimax hypothesis testing \rm is reserved for the test that provides robustness with respect to unknown probabilities of the hypothesis, i.e. $P(\mathcal H_0)$ and $P(\mathcal H_1)$, the term, \it robust hypothesis testing, \rm is widely used for a test that is robust against possible deviations from the nominal probability distributions under each hypothesis \cite{Levy,Hube65,kush08}.\\
In addition to robustness, which is crucial for the design of a decision maker, another important aspect is to include multiple decision makers (physical sensors) into the decision making process, because it is well known that if the events of interest are independent, the system error probability decreases exponentially with the number of sensors \cite{Chernoff}. Depending on the type of data, that is transmitted to the fusion center, a distributed detection network can either be centralized or decentralized, i.e. it is decentralized if the transmitted data is quantized and centralized otherwise. \cite{Tsi}.\\
Decentralized detection can be implemented with and without a fusion center network. Figure~\ref{fig1} illustrates a decentralized detection network without a fusion center (DDN-WoF), which is often preferred due to its robustness to single link failure, fairly simplified design procedure and the ease of scalability \cite{Cattivelli}. The fusion is established among sensors (decision makers) via exchange of information in an iterative way, e.g. via belief propagation. A decentralized detection network with a fusion center (DDN-WF), on the other hand, considers a fusion center where the decisions of each single decision maker is fused. An example of a DDN-WF with a parallel topology is illustrated in Figure~\ref{fig2}. Despite a physical lack of robustness, such as being prone to a single link failure, DDN-WF includes the previously introduced DDN-WoF as a special case in terms of error minimization, because DDN-WF allows joint optimization of local decision rules and fusion rules, whereas for DDN-WoF optimization is performed only over decision rules for a pre-determined fusion rule, see \cite[Chapter 3]{Varshney}. In this context, minimax decentralized detection fulfills two important requirements for any detection problem that is intended to be realized in practice: \it high detection accuracy \rm due to multiple sensors and \it reliability \rm due to the decentralized nature of detection as well as the robustness of hypothesis testing.\\
The design of robust decentralized hypothesis testing for with- and without fusion center networks has been studied for the first time by Geroniotis \cite{Ger1,Ger2,Ger3} and generalized thereafter \cite{Ger4}. Later in \cite{Poor1} it was proven that DDN-WF is minimax robust if the individual sensors employ minimax robust tests. Moreover, the authors formalized necessary conditions that need to be satisfied by the Bayesian cost assignment procedure for DDN-WoF. The results derived in \cite{Poor1} generalize the results of \cite{Ger2} to a network of more than two sensors and to more general cost functions. Recent studies in this area apply earlier results to scenarios with constraints such as power \cite{park}, communication rate, \cite{gul2}, or local optimality \cite{censoring}.
\begin{figure}[ttt]
  \centering
  \centerline{\includegraphics[height=55mm]{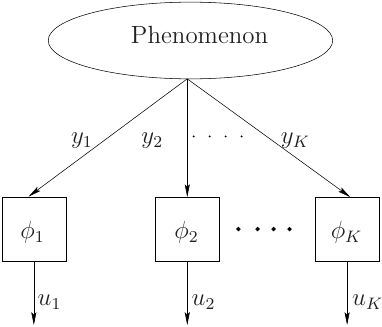}}
\caption{Distributed detection network without a fusion center for $K$ decision makers.\label{fig1}}
\end{figure}
Minimax decentralized hypothesis testing, on the other hand, can be designed as a special case of the Bayesian formulation of the decentralized hypothesis testing problem \cite{Tsi}. Such a design also imposes constraints on the choice of Bayesian costs and on the achievable performance, which parallels with the work of Veerevalli et al. \cite{Poor1}. Furthermore, there are universal losses in minimax decentralized hypothesis testing, and due to the differences in the design of DDN-WoF and DDN-WF. The availability of theoretical bounds makes it possible to know the limitations of the designs and possible gains due to conversion of decentralized networks without having any knowledge about the observation statistics. Although there have been several studies in robust hypothesis testing for a single sensor \cite{gul} and multiple sensors \cite{Poor1}, according to our best knowledge, the aforementioned theoretical bounds (the losses) as the implications of minimax hypothesis testing in the design of DDN-WoF and DDN-WF have never been published yet. \\
In this paper, decentralized detection networks in both the absence and the presence of a fusion center are studied. The constraints and losses resulting from minimax hypothesis testing in a single and a multi-sensor detection system are obtained. The maximum of the performance loss between DDN-WoF and DDN-WF are derived when both networks are minimax robust. Finally, the results are generalized. \\
The organization of the paper is as follows. In the following section, the problem definition is given. In Section~\ref{sec3} the constraints in the design of minimax DDN-WoF are determined. In Section~\ref{sec4} both for a single sensor as well as for multiple sensors, the maximum loss in detection performance due to minimax hypothesis testing is obtained. In Section~\ref{themaximumlossbetweenminimaxddnwofandddnwf0074} the maximum loss in detection performance between DDN-WoF and DDN-WF is derived when both networks are minimax robust. In Section~\ref{conclusions} the paper is concluded and the proofs are provided in Appendices~\ref{appendix1}-\ref{proofofequation00730}.
\begin{figure}[ttt]
  \centering
  \centerline{\includegraphics[width=65mm]{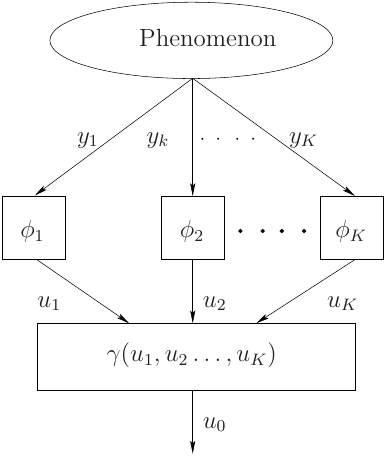}}
\caption{Distributed detection network with a fusion center for $K$ decision makers.\label{fig2}}
\end{figure}
\section{Problem Definition}\label{sec2}
Let $P_0$ and $P_1$ be two distinct probability measures on a measurable space $(\Omega,{\mathscr{A}})$, and for each decision maker $i$ consider the following binary hypothesis testing problem:
\begin{align}\label{eq1}
\mathcal{H}_0&: Y_i \sim P_0\nonumber\\
\mathcal{H}_1&: Y_i \sim P_1
\end{align}
where each $Y_i:\Omega\rightarrow \Omega$ is a continuous random variable corresponding to the observation $y_i\in \Omega$, see Figures~\ref{fig1},\ref{fig2}. Without loss of generality $\Omega$ can be any interval of real numbers or the whole real line\footnote{The random variables $Y_i$ can also be vector valued. This does not affect the results derived throughout the paper.}. To decide for the hypothesis $\mathcal{H}_j$, each sensor makes an observation $y_i$ and gives a decision $\phi_i \in \{0,1\}$. The decisions of all sensors $\boldsymbol{\phi}=\{\phi_1,...,\phi_n\}$ are either exchanged among sensors, or sent to the fusion center in order to give the final decision which is made based on a fusion rule $\gamma:\boldsymbol{\phi}\mapsto\{0,1\}$. The aim of decentralized hypothesis testing is to find the decision and fusion rules which minimize the Bayesian risk
\begin{equation}\label{eq2}
R=\sum_{i=0}^1\sum_{j=0}^1\pi_{j} C_{ij} P[\gamma(\boldsymbol{\phi})=i|{\mathcal H}_j]
\end{equation}
where $\pi_0=P(\mathcal{H}_0)$ and $\pi_1=P(\mathcal{H}_1)=1-\pi_0$ are the a priori probabilities and $C_{ij}\in\mathbb{R}_{\geq 0}$ are the costs of making a decision $i$ when hypothesis $j$ is true. To achieve minimum error probability it is reasonable to assign equal costs to error and detection probabilities, i.e. $C_{01}=C_{10}$ and $C_{00}=C_{11}$, where $C_{01}>C_{00}$ (since the risk is needed to be minimized). Together with this assumption defining the overall false alarm and miss detection probabilities as $P_{F_0}=P[\gamma(\boldsymbol{\phi})=1|{\mathcal H}_0]$ and $P_{M_0}=P[\gamma(\boldsymbol{\phi})=0|{\mathcal H}_1]$ leads to
\begin{align}\label{eq3}
R=&\pi_0 (P_{F_0}(C_{01}-C_{00})+P_{M_0}(C_{01}-C_{00}))\nonumber\\
+&P_{M_0}(C_{01}-C_{00})+C_{00}.
\end{align}
\subsection{Minimum error probability and minimax decision and fusion rules}\label{decisionfusionrules}
Depending on the practical application, it may or may not be possible to know the a priori probabilities. These two cases lead to two different optimization problems as explained in the following sections.
\subsubsection{Minimum error probability decision and fusion rules}\label{subsection1}
For various applications, such as digital communications or classification, it is reasonable to assign equal a priori probabilities, e.g. $\pi_0=1-\pi_1=1/2$, which results in decision and fusion rules that solve
\begin{equation}\label{eq4}
(\boldsymbol{\phi}_0,\gamma_0)=\arg \min_{\boldsymbol{\phi},\gamma}{R(\boldsymbol{\phi},\gamma,\cdot)},
\end{equation}
where
\begin{align}\label{eq5}
{R(\boldsymbol{\phi},\gamma,\cdot)}=&\frac{1}{2}(C_{01}-C_{00})(P_{F_0}(\boldsymbol{\phi},\gamma,P_0)\nonumber\\
+&P_{M_0}(\boldsymbol{\phi},\gamma,P_1))+C_{00}.
\end{align}
Here, $\boldsymbol{\phi}_0$ and $\gamma_0$ are called the minimum error decision and fusion rules.
\subsubsection{Minimax decision and fusion rules}\label{subsection2}
In many other applications, including radar and sonar the a priori probabilities are not known and it is preferable to design the test based on the least favorable priors. This involves solving
\begin{align}\label{eq6}
(\boldsymbol{\phi}_r,\gamma_r,{\pi_0}_r)=&\arg \underset{\boldsymbol{\phi},\gamma}{\rm min\it}\,\,\underset {\pi_0}{\rm max\it}\,\,{R(\pi_0,\boldsymbol{\phi},\gamma,\cdot)}\nonumber\\
=&\arg \underset{\pi_0}{\mathrm{max}}\,\,\underset {\boldsymbol{\phi},\gamma}{\rm min\it}\,\,{R(\pi_0,\boldsymbol{\phi},\gamma,\cdot)}
\end{align}
where
\begin{align}\label{eq7}
{R(\pi_0,\boldsymbol{\phi},\gamma,\cdot)}=&\pi_0(C_{01}-C_{00})\left(P_{F_0}(\boldsymbol{\phi},\gamma,P_0)-P_{M_0}(\boldsymbol{\phi},\gamma,P_1)\right)\nonumber\\
+&(C_{01}-C_{00})P_{M_0}(\phi,\gamma,P_1)+C_{00}.
\end{align}
In \eqref{eq6} performing the maximization first leads to
\begin{equation}\label{eq8}
\frac{\partial R(\pi_0,\boldsymbol{\phi},\gamma,\cdot)}{\partial\pi_0}=0\Longrightarrow P_{F_0}(\boldsymbol{\phi},\gamma,P_0)=P_{M_0}(\boldsymbol{\phi},\gamma,P_1)
\end{equation}
Inserting \eqref{eq8} in \eqref{eq6} the problem definition can be rewritten as
\begin{align}\label{eq9}
(\boldsymbol{\phi}_r,\gamma_r) &= \arg \min_{\boldsymbol{\phi},\gamma} P_{M_0}(\boldsymbol{\phi},\gamma,P_1)\nonumber\\
 &\mathrm{s.t.}\quad P_{F_0}(\boldsymbol{\phi},\gamma,P_0)=P_{M_0}(\boldsymbol{\phi},\gamma,P_1).
\end{align}
Note that $(\boldsymbol{\phi}_r,\gamma_r,{\pi_0}_r)$ constitutes a saddle value for \eqref{eq6}, therefore a certain level of detection performance can always be maintained \cite{Levy}, i.e. inserting \eqref{eq8} in \eqref{eq7} makes $R$ independent of $\pi_0$. The solution of \eqref{eq9} results in minimax decision and fusion rules.

\subsection{Assumptions}\label{assumptions}
The following three assumptions will be valid throughout the paper.

\begin{assumption}\label{assumption1}
The random variables $Y_i$, which correspond to the observations $y_i$, are identically distributed and mutually independent.
\end{assumption}

\begin{assumption}\label{assumption2}
The decision rules are identical, i.e. $\phi_i(Y)=\phi_j(Y)$ almost surely for all pairs $(i,j)$.
\end{assumption}

\begin{assumption}\label{assumption3}
The density of the likelihood ratio function $l=d P_1/d P_0$ does not have a point mass under each hypothesis, i.e. $P[l(Y_i)=c|\mathcal{H}_0]=P[l(Y_i)=c|\mathcal{H}_1]=0$ for all constant $c$ and decision maker $i$.
\end{assumption}

The first assumption is basically needed for mathematical tractability. It is also a fairly valid assumption if the sensors are geographically disperse. The second assumption is considered both for DDN-WoF and DDN-WF for the sake of fairness of comparisons. Another motivation for the second assumption is due to the result of minimax modeling of DDN-WoF, which will be shown in Section~\ref{sec_a}. Note that identical local decision makers are not always optimum, see counterexamples in \cite{Kantor}, but they often result in little or no loss of performance \cite{Varshney} and they are asymptotically optimum \cite{Tsi2}. The last assumption is needed to make the Receiving Operating Characteristic (ROC) curve continuous. This property will later be used in the derivations.

\subsection{Different solutions of the defined optimization problems}\label{solutionsforproblems}
The optimization problems defined by \eqref{eq4} in Section~\ref{subsection1} and by \eqref{eq9} in Section~\ref{subsection2}  can be solved differently by considering different network configurations as detailed in the following.

\subsubsection{Singe Sensor}\label{singlesensordefinition}
If there is a single sensor in the sensor network, the fusion rule $\gamma$ is undefined. There is a single decision rule $\boldsymbol{\phi}_r:=\phi_r$ or $\boldsymbol{\phi}_0:=\phi_0$ and the overall false alarm and miss detection probabilities, $P_{F_0}$ and $P_{M_0}$ are replaced by $P_F$ and $P_M$, respectively.

\subsubsection{Multi-sensor DDN-WoF}\label{multisensordefinition}
For multiple sensor decentralized detection network without a fusion center the solution of \eqref{eq4} and \eqref{eq9} is performed only over the decision rules $\boldsymbol{\phi}$. The fusion rule $\gamma$ is fixed a-priori, see \cite[Chapter 3]{Varshney}. The main advantage of such a design is that the optimization over only decision rules is much simpler than an optimization over both decision and fusion rules. Moreover, there are readily accepted fusion rules such as the majority voting rule, which is known to be quite powerful. If all sensors are identical, and the fusion is performed iteratively via exchange of decisions, then the result converges to the decision of the majority. Therefore, for $K$ decision makers, the fusion rule can be chosen as the majority voting rule defined by
 \begin{equation}\label{eq10}
\gamma^K= \begin{cases} 0, &  \mbox{ $\sum_i^K \phi_i<\frac{K}{2}$} \\ \kappa, &  \mbox{ $\sum_i^K \phi_i=\frac{K}{2}$}\\ 1, &  \mbox{ $\sum_i^K \phi_i>\frac{K}{2}$} \end{cases}
\end{equation}
where $\kappa\in\{0,1\}$ is a Bernoulli random variable (r.v.) with success probability $P(\kappa=1)=1/2$ \cite{frasca}. In the sequel, for DDN-WoF (as well as for DDN-WF), false alarm and miss detection probabilities of each sensor are denoted by $P_F=P_{F_i}$ and $P_M=P_{M_i}$ for all $i$, respectively, as a consequence of Assumptions~\ref{assumption1},\ref{assumption2}.

\subsubsection{Multi-sensor DDN-WF}\label{multisensordefinition2}
Decentralized detection network with a fusion center provides a solution to \eqref{eq4} and \eqref{eq9} over both the decision and the fusion rules, still considering all the assumptions given in Section~\ref{assumptions}. In particular, Assumptions~\ref{assumption1},\ref{assumption2} imply that $\sum_i \phi_i\sim \mbox {Binomial}$ is the sufficient statistic and the optimum fusion rule is called K-out-of-N fusion rule \cite{Tsi,Varshney}:
\begin{equation*}
\gamma^K_{k} = \begin{cases} 0, &  \mbox{ $\sum_i^K \phi_i\leq k$} \\ 1, &  \mbox{ $\sum_i^K \phi_i>k$}. \end{cases}
\end{equation*}
Hence, the overall false alarm and miss detection probabilities are given by
\begin{align}\label{eq11}
P_{F_0}(P_F;K,k)=&1-B(k;K,P_F)\nonumber\\
=& 1-\sum_{i=0}^{\lfloor k\rfloor}\binom{K}{i}P_F^i(1-P_F)^{K-i}
\end{align}
and
\begin{align}\label{eq12}
P_{M_0}(P_M;K,k)=&B(k;K,1-P_M)\nonumber\\
=& \sum_{i=0}^{\lfloor k\rfloor}\binom{K}{i}(1-P_M)^i P_M^{K-i}
\end{align}
where $B$ is a binomial cumulative distribution function with at most $k$ successes out of $K$ total trials. Equations \eqref{eq11}, \eqref{eq12} imply that the optimization over $\gamma$ in \eqref{eq4} and \eqref{eq9} can be performed over integers $k$. Additionally, the choice of $k=K/2$ corresponds to the error probabilities resulting from DDN-WoF.

\subsection{Objectives}
Having defined two different optimization problems in Section~\ref{decisionfusionrules} and different solutions in Section~\ref{solutionsforproblems} considering the assumptions in Section~\ref{assumptions}, next, the aim is to define the objectives:
\begin{itemize}
\item The first aim is to determine the restrictions in the design of minimax DDN-WoF, in particular, the constraints on the choice of Bayesian costs and on the achievable performance.
\item The second aim is to find the maximum (over all probability distributions on $({\Omega},\mathscr{A})$) performance loss for DDN-WoF if the minimax decision rule is considered instead of the minimum error probability decision rule.
\item The third aim is to find the maximum (over all probability distributions on $({\Omega},\mathscr{A})$) performance loss between minimax DDN-WoF and minimax DDN-WF.
\end{itemize}
The solutions of these three problems are addressed in the following sections in the same order as given above.


\section{Constraints in the Design of Minimax DDN-WoF}\label{sec3}
Minimax formulation of the decentralized detection problem imposes constraints on the choice of the Bayesian costs as well as on the achievable performance since the problem is a constraint type of optimization.\footnote{The choice of the more general costs functions, see for example \cite{bchen}, may result in different costs assignments for minimax robustness. Throughout this paper only the Bayesian formulation of the hypothesis testing problem is considered.} The following sections address these two issues.

\subsection{Constraints on the System Design}\label{sec_a}
Consider the Bayesian formulation of the risk function
\begin{equation}\label{eq13}
R=\sum_{i_1,\ldots, i_K,j}\pi_{j} C_{i_1,\ldots, i_K, j}P_{j}[\phi_{1}=i_1,\ldots, \phi_{K}=i_K]
\end{equation}
where $C_{(\cdot)}$ is the cost of making decisions $\phi_{1}=i_1,\ldots,\phi_{K}=i_K$ when the true hypothesis is $\mathcal{H}_j$. Minimax decision rules $\boldsymbol{\phi}_r$ should minimize $R$ subject to the constraint $P_{F_0}=P_{M_0}$. A person by person optimum (PBPO) solution that minimizes \eqref{eq13} is known by Tsitlikis \cite{Tsi} and Varshney \cite{Varshney} to be the likelihood ratio tests. The results in \cite{Varshney} may be preferable over \cite{Tsi}, since \cite{Varshney} treats DDN-WoF and DDN-WF separately and specifies the likelihood ratio tests explicitly. Such a specification requires tedious steps of derivations and for every $K$, the derivations should be repeated from the beginning. Therefore, the results are not generic for every number of sensors and in such a case, a reasonable approach, which will be followed shortly, is to consider a specific case and generalize the results thereafter. When $K=2$, \cite{Varshney} suggests that if Assumption~\ref{assumption1} holds, the likelihood ratio tests for two sensors are given by
\begin{equation}\label{eq14}
\phi_1:\frac{p_1(y_1)}{p_0(y_1)}\,\stackrel{\phi_1=1}{\underset{\phi_1=0}{\lesseqgtr}}\frac{\pi_0\int_{E}p_0(y_2)\left[C_a+P[\phi_2=0|y_2]C_b\right]\mathrm{d}y_2}{\pi_1\int_{E}p_0(y_2)\left[C_c+P[\phi_2=0|y_2]C_d\right]\mathrm{d}y_2},
\end{equation}
and $\phi_2$, which is obtained similarly by exchanging $(y_2,\phi_2)$ with $(y_1,\phi_1)$ in \eqref{eq14}, where
\begin{align*}
C_a=C_{110}-C_{010},\quad  C_b=C_{100}-C_{000}+C_{010}-C_{110},  \nonumber\\
C_c=C_{011}-C_{111},\quad  C_d=C_{001}-C_{101}+C_{111}-C_{011}. \nonumber
\end{align*}
The DDN-WoF shown in Figure~\ref{fig1} does not allow exchanging observations among the decision makers. In contrast, the optimum decision rules given by \eqref{eq14} requires the observation of the other sensor, e.g. $\phi_1$ needs $y_2$ ,unless $C_b=0\wedge C_d=0$. Hence, the choice of costs
\begin{equation}\label{eq15}
{\mathcal{C}}_2=\{C_{i_1i_2j}:C_b=0\wedge C_d=0\}
\end{equation}
guarantees a possible minimax test. As a result, both decision makers own the same thresholds (\eqref{eq14} follows from \eqref{eq15}), which leads to Assumption~\ref{assumption2}. For $K>2$, it is expected that there are costs, which couple the sensor decisions and some other costs, which have no effect in coupling. This generalization suggests that the coupling terms must be set to zero and the other parameters should be determined accordingly in order to guarantee minimax robustness. The choice of majority voting rule and identical sensor decisions is a consequence of the analysis above. The identical sensor decisions can further be restricted with the following proposition.
\begin{prop}\label{prop1}
For identical sensor decisions, $P_{F_0}=P_{M_0}$ if and only if all the local decision rules result in $\theta=P_{F}(\phi,P_0)=P_{M}(\phi,P_1)$, with a suitable choice of $(C_a,C_b)$ and $(\pi_0,\pi_1)$.
\end{prop}
A short proof of Proposition~\ref{prop1} is given in Appendix~\ref{appendix1}. In conclusion, \eqref{eq4} has a solution for DDN-WoF with the majority voting rule and identical local sensor decisions, which result in $\theta=P_M=P_F$ for every sensor. In the following sections, DDN-WoF will be referred to this setting.

\subsection{Constraints on the Achievable Performance}\label{sec_b}
In addition to the constraint on the selection of the costs to achieve a robust test, there is a similar restriction, in terms of performance, when a new sensor is added to the sensor network. The following proposition states this claim.
\begin{prop}\label{prop2}
When majority voting rule is combined with identical local decision makers, adding one more sensor to the network of $2i-1$, $i\geq 1$ sensors does not improve the detection performance.
\end{prop}
In Appendix~\ref{appendix2}, Proposition~\ref{prop2} is proven by showing that the false alarm probability for $2K-1$ sensors
\begin{equation*}
P_{F_0}^{2K-1}=\sum_{i=K}^{2K-1}\binom {2K-1} i{P_F}^i(1-{P_F})^{2K-1-i}
\end{equation*}
is equal to the false alarm probability for $2K$ sensors
\begin{align}\label{eq16}
P_{F_0}^{2K}&=\sum_{i=K+1}^{2K}\binom {2K} i {P_F}^i(1-P_F)^{2K-i}\nonumber\\
&+\frac{1}{2}\binom {2K} K {P_F}^K(1-P_F)^K
\end{align}
where the second term in \eqref{eq16} comes from the randomization in \eqref{eq10}. Since $P_{F_0}=P_{M_0}$, see Appendix~\ref{appendix2}, it follows that $P_{M_0}^{2K-1}=P_{M_0}^{2K}$.

\section{The Maximum Loss due to Minimax Decision Making in DDN-WoF}\label{sec4}
From the previous section, there are two main conclusions: first is the setting of DDN-WoF, i.e. majority voting rule together with identical sensor decisions leading to $\theta=P_M=P_F$, and second is the fact that increasing the number of sensors from odd numbers to even numbers does not improve the detection performance. Hence, without loss of generality, only the odd numbered sensors can/will be considered with the given DDN-WoF setting for further analysis.\\
Before proceeding with the derivations, some sets and functions are needed to be defined. The main idea is to perform the optimization over the receiver operating characteristics (ROC)s which uniquely characterizes all possible false alarm and miss detection pairs resulting from the decision rules ${\phi}_r$ and ${\phi}_0$. Formally,
\begin{equation*}
\mathcal{P}_\mathcal{F}\times \mathcal{P}_\mathcal{M}=\{(P_F(\phi(t),P_0),P_M(\phi(t),P_1)):\forall t\in[0,\infty]\}
\end{equation*}
defines the ROC, where $\mathcal{P}_\mathcal{F}$ is the set of all $P_F$s and $\mathcal{P}_\mathcal{M}$ is the set of all corresponding $P_M$s. Let $r_t:[0,1]\rightarrow [0,1]$ be a function with the mapping $\mathcal{P}_\mathcal{F} \overset{r_t}{\mapsto} \mathcal{P}_\mathcal{M}$, and $\phi(t)$ is the deterministic likelihood ratio test with threshold $t$. The following Lemma is necessary for further analysis:
\begin{lem}\label{lemma1}
Let Assumption~\ref{assumption3} holds. The function $r_t$ is bounded on $[0,1]^2$, passes through $(P_F,P_M)=(1,0)$ and $(P_F,P_M)=(0,1)$, and is continuous, and convex.
\end{lem}
A proof of Lemma~\ref{lemma1} can be found in Appendix~\ref{appendix3}.
\begin{exmp}\label{example1}
For mean shifted Gaussian distributions $P_0$ and $P_1$ each having a variance $\sigma^2=1$ and, a difference in mean $d=1$, the function $r_t$ is given by:
\begin{equation*}
P_M=r_t(P_F)=F\left(F^{-1}(1-P_F)-1\right)
\end{equation*}
\end{exmp}
\noindent where $F$ is the cumulative distribution function of the standard Gaussian distribution. Lemma~\ref{lemma1} implies that every $r_t$ lies below the line $P_M=1-P_F$. Hence, all simple hypothesis testing problems with the nominal distributions $P_0$ and $P_1$ can be identified with the lower triangle of $[0,1]^2$ for performance evaluation. Let $P_0^\theta$ and $P_1^\theta$ denote two distinct probability distributions on $(\Omega,\mathscr{A})$ which satisfies Assumption~\ref{assumption3} and for which a likelihood ratio test $\phi$ yields an error probability of $\theta$. Then, the set of all such pairs of distributions is
\begin{equation*}
({\mathcal P}_0\times{\mathcal P}_1)(\theta)=\{(P_0^\theta,P_1^\theta)|\exists \phi: \theta= P_F(\phi,P_0^\theta)=P_M(\phi,P_1^\theta)\}.
\end{equation*}
Hence, varying $\theta\in[0,1/2)$, the set $({\mathcal P}_0\times{\mathcal P}_1)(\theta)$ covers all distinct pairs of distributions on $(\Omega,\mathscr{A})$ satisfying Assumption~\ref{assumption3}. Due to the convexity of $r_t$, see Lemma~\ref{lemma1}, if any pair of distributions $(P_0^\theta,P_1^\theta)$ belong to $({\mathcal P}_0\times{\mathcal P}_1)(\theta)$, then the function $r_t$ lies in the butterfly shaped area
\begin{align*}
\mathcal{B}_\theta=\{\{(P_F,P_M):P_M\geq l_1(P_F)\}\cap \{(P_F,P_M):P_M\leq l_2(P_F)\}\}\nonumber\\
\cup \{\{(P_F,P_M):P_M\leq l_1(P_F)\}\cap \{(P_F,P_M):P_M\geq l_2(P_F)\}\}
\end{align*}
defined by the intersection of two lines,
\begin{align}\label{eq17}
l_1= \{(P_F,P_M):P_M=&\hat\theta(1-P_F)\},\nonumber\\
l_2= \{(P_F,P_M):P_M=&(1-P_F/\hat\theta)\}
\end{align}
on $[0,1]^2$, where $\hat\theta=\theta/(1-\theta)$. Notice that $l_1$ is the inverse function of $l_2$. An example of $\mathcal{B}_\theta$, for $\theta\approx0.309$, together with the lines $l_1$ and $l_2$, and $P_M=r_t(P_F)$ given by Example~\ref{example1} are illustrated in Figure~\ref{fig3}. The function $h_{k}^K$ will be defined later.\\
\begin{figure}[ttt]
  \centering
  \centerline{\includegraphics[width=9.0cm]{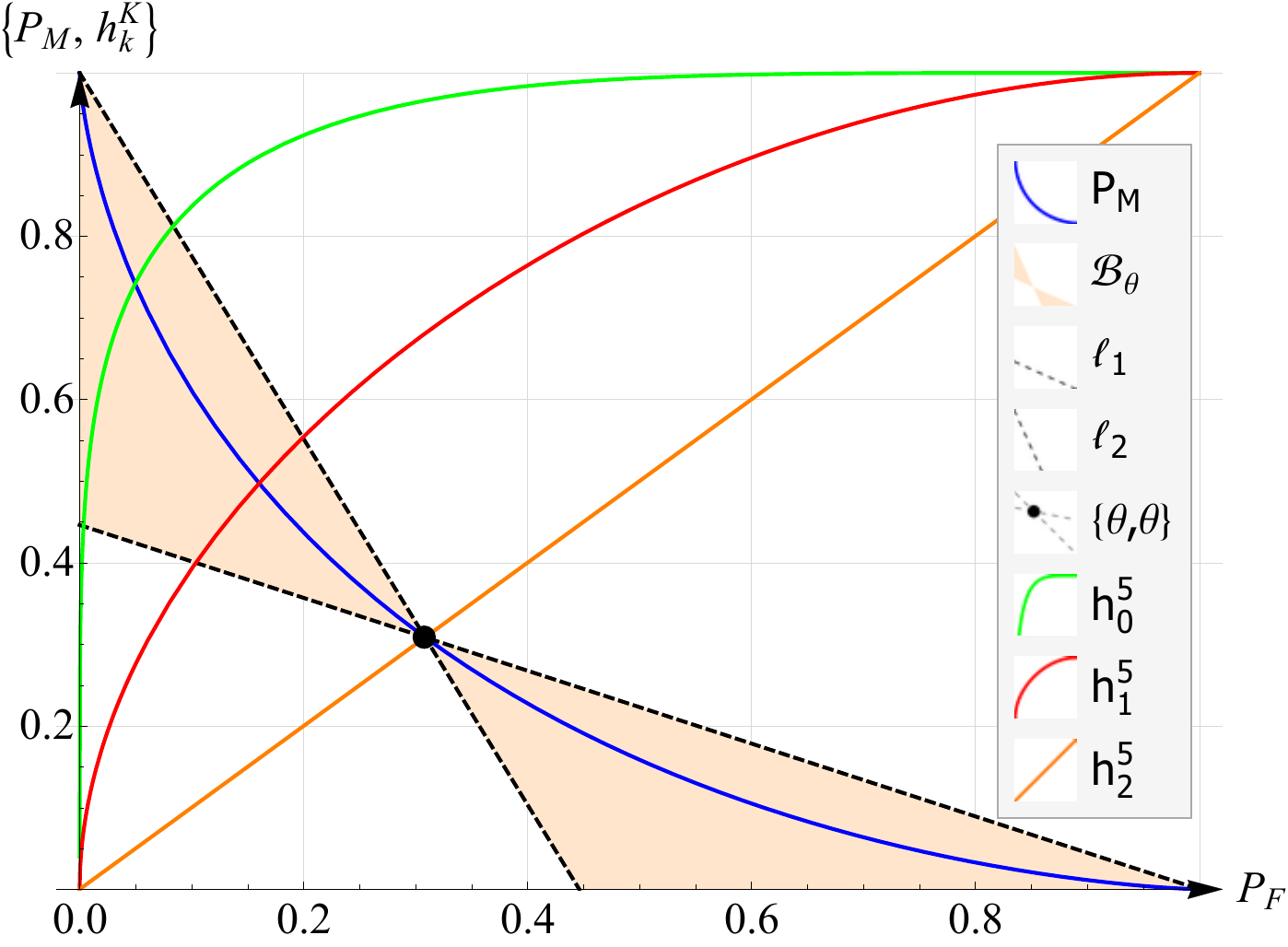}}
\caption{An example of an ROC curve together with the fusion function $h_{k}^K$ for $K=5$ and $k\in\{0,1,2\}$.\label{fig3}}
\end{figure}Next, the aim is to find the maximum loss of detection performance between minimax and minimum error probability decision rules, i.e. between $\boldsymbol{\phi}_r$ and $\boldsymbol{\phi}_0$. The fusion rule $\gamma^K$ is fixed as the majority voting rule for both decision strategies.

\subsection{Single Sensor Case}\label{singlesensorcase00731}
Assume that DDN-WoF has a single decision maker, cf. Section~\ref{singlesensordefinition}, which uses the minimax decision rule $\phi_r$, in comparison to the minimum error probability decision rule $\phi_0$. The loss between the error probabilities of $\phi_r$ and $\phi_0$ can then be found from \eqref{eq5} and \eqref{eq7} with \eqref{eq9}:
\begin{align*}
L(\phi_r,\phi_0,P_0^\theta,P_1^\theta)&=R(\pi_0,\phi_r,\cdot)-R(\phi_0,\cdot)\nonumber\\
&=(C_{01}-C_{00})P_{M}(\phi_r,P_1^\theta)\nonumber\\
&-\frac{1}{2}(C_{01}-C_{00})\left(P_{F}(\phi_0,P_0^\theta)+P_{M}(\phi_0,P_1^\theta)\right)\nonumber\\
&=\frac{1}{2}(C_{01}-C_{00})(2P_{M}(\phi_r,P_1^\theta)\nonumber\\
&-(P_{F}(\phi_0,P_0^\theta)+P_{M}(\phi_0,P_1^\theta))).\nonumber
\end{align*}
Now, the aim is to answer the following problem:
\begin{prb}\label{question1}
What is the maximum of $L$ over all $(P_0^\theta,P_1^\theta)\in ({\mathcal P}_0\times{\mathcal P}_1)(\theta)$ in terms of $\theta$, and in general?
\end{prb}
\noindent The error probability resulting from $\phi_r$ is $\theta=P_{F}(\phi_r,P_0^\theta)=P_{M}(\phi_r,P_1^\theta)$. Minimizing the average error probability resulting from $\phi_0$, hence, maximizing the loss function $L$, is equivalent to finding a $(P_F,P_M)$ on $\mathcal{B}_\theta$ such that $P_F+P_M$ is minimum. Notice that the butterfly $\mathcal{B}_\theta$ is symmetric with respect to $P_M=P_F$. Thus, only the upper part $(P_M\geq\theta)$ or the lower part $(P_M\leq \theta)$ of $\mathcal{B}_\theta$ can be considered, though both of them result in the same error probability. Let $$l_1^*=\{l_1:P_M\geq \theta\}\quad\mbox {and}\quad l_2^*=\{l_2:P_M\leq \theta\}$$ be the line segments of the lines $l_1$ and $l_2$, respectively. Then, the minimum error probability is achieved either on $l_1^*$ or on $l_2^*$. This is because, the points on $l_1^*$ have the lowest $P_M$ for every $P_F$ on upper $\mathcal{B}_\theta$, and the points on $l_2^*$ have the lowest $P_F$ for every $P_M$ on the lower $\mathcal{B}_\theta$. The average error probability increases on $l_1^*$ and decreases on $l_2^*$. Hence, either $(P_F,P_M)=(0,\hat\theta)$, cf.~$l_1(P_F:=0)$ or $(P_F,P_M)=(\hat\theta,0)$ cf.~$l_2(P_F:=\hat\theta)$ minimizes the average error probability for $\phi_0$. Since neither $(0,\hat\theta)$ nor $(\hat\theta,0)$ is achievable as $r_t$ must be continuous with $r_t(0)=1$ and $r_t(1)=0$ it follows that
\begin{align*}
\underset {(P_0^\theta,P_1^\theta)\in {(\mathcal P}_0\times{\mathcal P}_1)(\theta)}{\sup}\,\,&{ L}(\phi_r,\phi_0,P_0^\theta,P_1^\theta)=(C_{01}-C_{00})(\theta-{\hat\theta}/{2})\nonumber\\
&=\frac{(C_{01}-C_{00})}{2}\frac{\theta (1-2\theta)}{1-\theta}
\end{align*}
and
\begin{align}\label{eq18}
\underset{\theta\in[0,1/2)}{\max}\underset {(P_0^\theta,P_1^\theta)\in {(\mathcal P}_0\times{\mathcal P}_1)(\theta)}{\sup}&\,\,{L}(\phi_r,\phi_0,P_0^\theta,P_1^\theta)\nonumber\\
&=(C_{01}-C_{00})\frac{3-2\sqrt{2}}{2}.
\end{align}
Equation \eqref{eq18} suggests that if the costs are chosen as $C_{01}-C_{00}=1$, the loss between $\phi_r$ and $\phi_0$ is upper bounded by $\approx 0.086$.

\subsection{Multiple Sensor Case}\label{multiplesensorcase00732}
Consider the distributed detection network without a fusion center as illustrated in Figure~\ref{fig1}. Let the total number of decision makers not be restricted to $1$, i.e. $K\geq 1$. In this case the loss in detection performance between minimax DDN-WoF and minimum error probability DDN-WoF can be found from \eqref{eq5}, and \eqref{eq7} with \eqref{eq9}:
\begin{align*}
L^K(\phi_r,\phi_0,\gamma^K,P_0^\theta,P_1^\theta)=&R(\pi_0,\phi_r,\gamma^K,\cdot)-R(\phi_0,\gamma^K,\cdot)\nonumber\\
=&\frac{1}{2}(C_{01}-C_{00})(2P_{M_0}(\phi_r,\gamma^K,P_1^\theta)\nonumber\\
-&(P_{F_0}(\phi_0,\gamma^K,P_0^\theta)+P_{M_0}(\phi_0,\gamma^K,P_1^\theta))).\nonumber
\end{align*}
The aim is to address the following problem:
\begin{prb}\label{question2}
What is the maximum of $L$ over all $(P_0^\theta,P_1^\theta)\in ({\mathcal P}_0\times{\mathcal P}_1)(\theta)$ in terms of $\theta$, and in general?
\end{prb}
\noindent Problem~\ref{question2} is a redefinition of Problem~\ref{question1} for multiple sensors. Proposition~\ref{prop1} and Equation~\eqref{eq12} indicate that minimax solution of DDN-WoF leads to $$P_{M_0}(\phi_r,\gamma^K,P_1^\theta)=P_{F_0}(\phi_r,\gamma^K,P_0^\theta)=B(K/2;K,1-\theta).$$ For the minimum error probability decision rule $\boldsymbol{\phi}_0$, it is assumed that all individual decision makers $\phi_{0,i}$ are identical, cf.~Section\ref{assumptions}. Hence, $P_F=P_{F_i}$ and $P_M=P_{F_i}$ for each decision maker, $i$, and due to the analyzes in the previous section, error minimizing $(P_F,P_M)$ lies either on the line segment $l_1^*$ or on $l_2^*$. Without loss of generality either of them can be considered, because $\mathcal{B}_\theta$ is symmetric with respect to $P_M=P_F$. Considering the former choice, let $x\in[1,\infty]$ be a free parameter. Then, $(P_F,P_M)=(\theta/x,(\theta(x-\theta))/(x(1-\theta)))$ characterizes $l_1^*$ completely. This indicates that some choice of $x:=x_{\min}$ minimizes the error probability resulting from $(\boldsymbol{\phi}_0,\gamma^K)$ and maximizes the loss function
\begin{align*}
&L^K_x (\theta)=(C_{01}-C_{00})\Bigg(B(K/2;K,1-\theta)\nonumber\\
&-\frac{1}{2}\Big(B(K/2;K,1-{\theta}/{x})+B\Big(K/2;K,1-\frac{\theta(x-\theta)}{x(1-\theta)}\Big)\Big)\Bigg).
\end{align*}
Figure~\ref{fig4} and \ref{fig5} illustrate $x_{\min}$ and the maximum of ${L}^K_x$, i.e. ${L}^K_{x_{\min}}$, respectively, over $\theta\in(0,1/2)$ and for every $K\in\{3,5,7,9\}$, where $(C_{01}-C_{00})=1$.
\begin{figure}[ttt]
  \centering
  \centerline{\includegraphics[width=9.2cm]{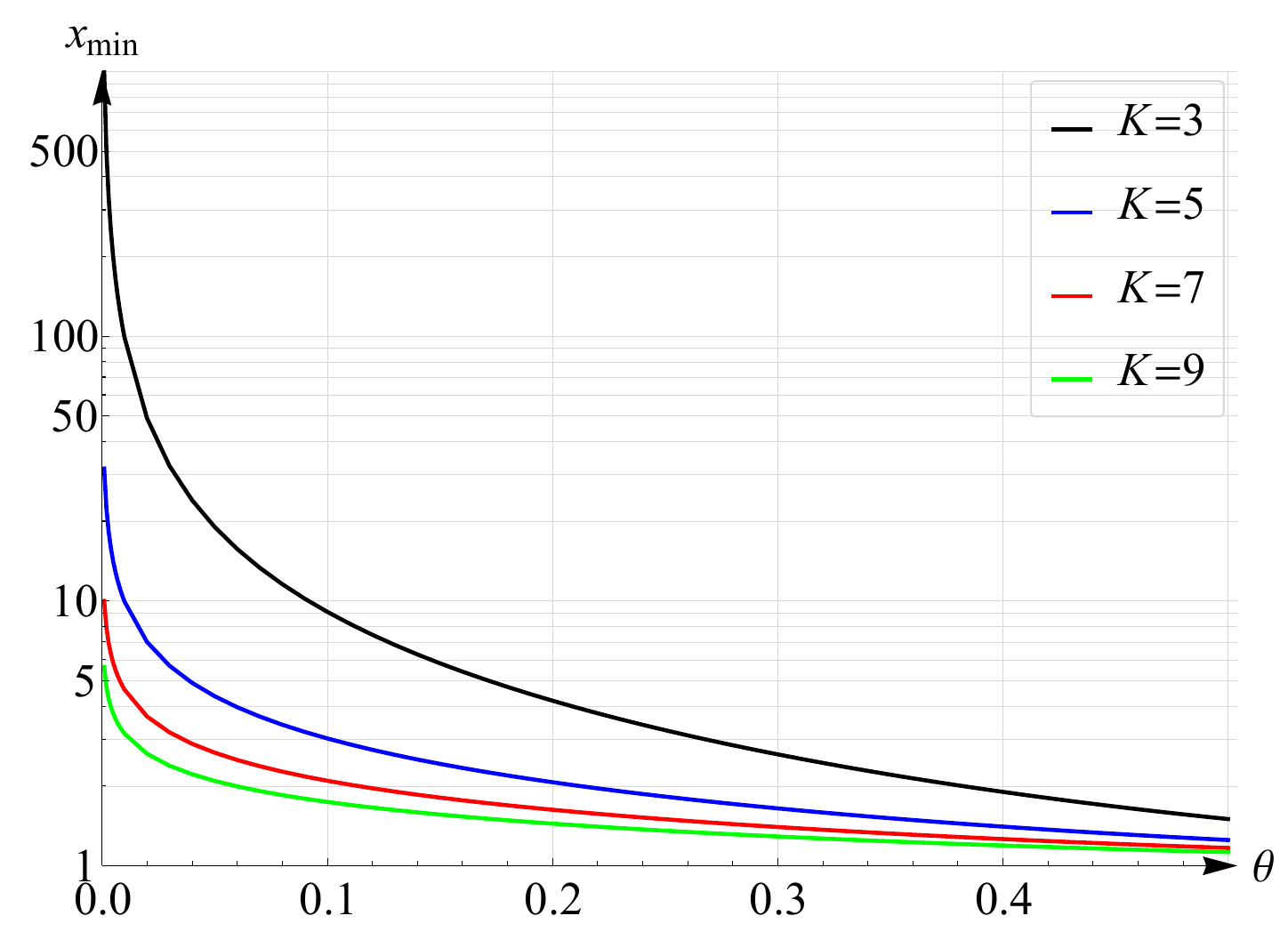}}
\caption{The parameter $x_\mathrm{min}$, which results in the maximum performance loss due to minimax decision making in DDN-WoF for all performance measures $\theta\in(0,1/2)$ and various number of sensors $K$.\label{fig4}}
\end{figure}
Notice that ${L}^1_{x_{\min}}$ is the same with the maximum of the loss $L$ found in the previous section.
\begin{figure}[ttt]
  \centering
  \centerline{\includegraphics[width=9.0cm]{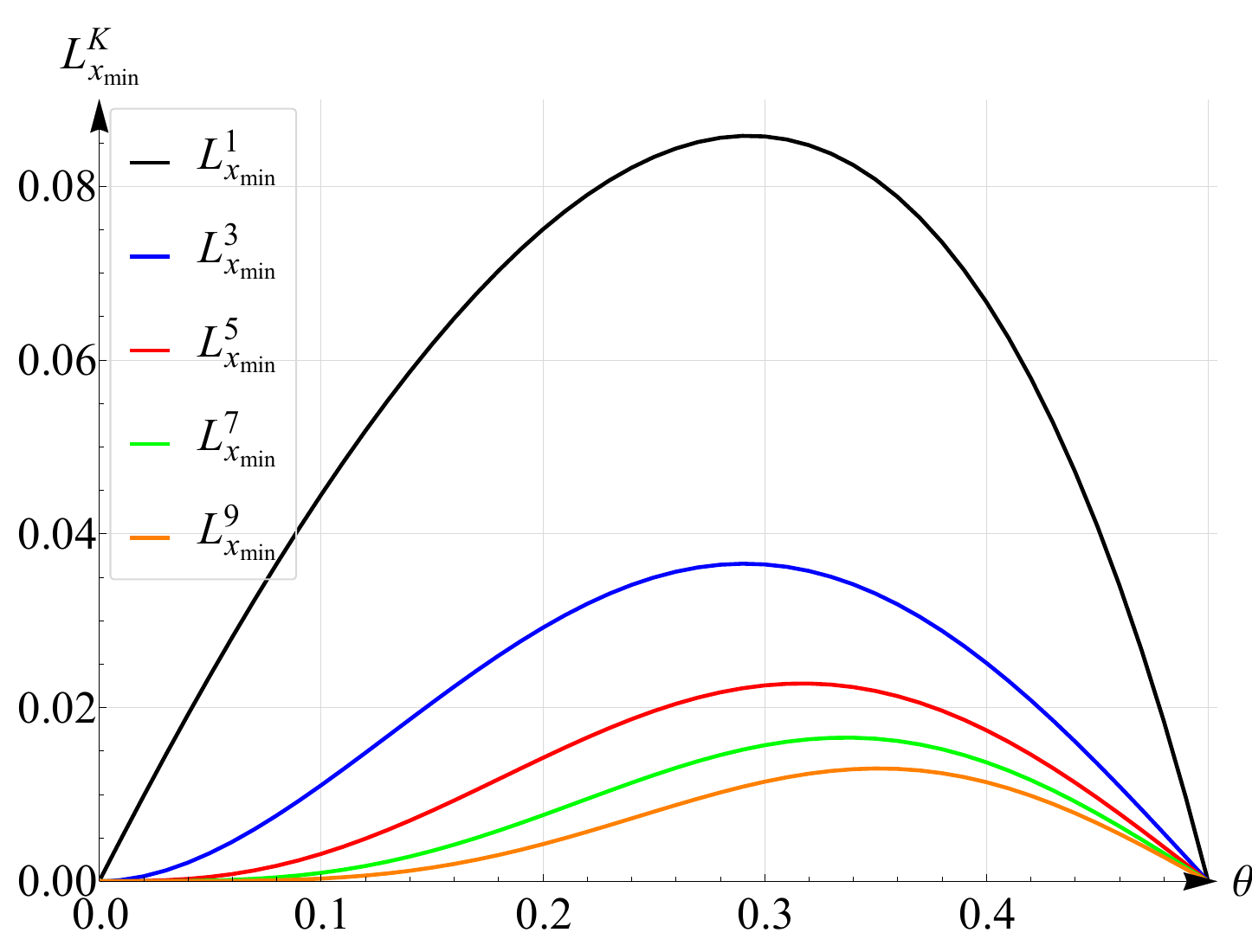}}
\caption{The maximum performance loss due to minimax decision making in DDN-WoF for all performance measures $\theta\in(0,1/2)$ and various number of sensors $K$.\label{fig5}}
\end{figure}
As $x\rightarrow \infty$, which corresponds to the single sensor case, the loss function becomes
\begin{equation*}
L^K_\infty (\theta)=B(K/2;K,1-\theta)-\frac{1}{2}B\left(K/2;K,\frac{1-2\theta}{1-\theta}\right).
\end{equation*}
\begin{prop}\label{prop3}
$L^K_\infty$ is almost everywhere negative on $\theta\in(0,1/2)$ for sufficiently large $K$.
\end{prop}
\noindent Proposition~\ref{prop3} is proven in Appendix~\ref{appendix4}. Negativity of $L^K_\infty$ indicates that single sensor optimum decision rules, which maximize the loss function are no more optimum for multiple sensor case for every $\theta\in(0,1/2)$. Because ${L}^K_{x_{\min}}$ is lower bounded by $0$.
\begin{rem}\label{remark1}
The choice of $\boldsymbol{\phi}_r$ at local decision makers do not only lead to a minimax test but also to an asymptotically optimal test, i.e. to the fastest decay rate of $P_{F_0}$ and $P_{M_0}$ as $K \rightarrow\infty$. This is due to the fact that the fastest exponential decay rate of the error probability amounts to equal rates of decreases for the false alarm and miss detection probabilities \cite[pp.74-82]{Levy}. $P_{F_0}$ and $P_{M_0}$ share the same polynomial function $f$ for the majority voting rule, cf. Appendix~\ref{appendix1}. Hence, they also share the same exponential decay rate whenever $\theta=P_M=P_F$. This is equivalent to $x_{\min}\rightarrow 1$ as $K\rightarrow \infty$, see Figure~\ref{fig4}.
\end{rem}
\section{The Maximum Loss Between Minimax DDN-WoF and DDN-WF}\label{themaximumlossbetweenminimaxddnwofandddnwf0074}
Consider the distributed detection networks with and without a fusion center illustrated in Figures~\ref{fig1} and~\ref{fig2}. For both DDN-WoF and DDN-WF, each local decision maker observes a phenomenon which is characterized by the distribution functions $(P_0^\theta,P_1^\theta)$ that belong to ${(\mathcal{P}_0\times \mathcal{P}_1)(\theta)}$. In either case the sensors are known to employ likelihood ratio tests and $\theta<1/2$. Both sensor networks have $K$ decision makers, where $K$ is assumed to be odd, cf. Proposition~\ref{prop2}, Section~\ref{sec4}, and the fusion center is assumed not to collect any observation.\\
As shown in the previous sections, the DDN-WoF solution to minimax hypothesis testing problem suggests that the loss of performance due to minimax decision making is small, see Figure~\ref{fig5}, and is asymptotically optimum cf. Remark~\ref{remark1}. In addition to the statistically satisfactory properties, DDN-WoF is also physically robust, making it appealing for real time applications. A major drawback of DDN-WoF compared to DDN-WF is a possible loss of performance, as DDN-WF allows joint optimization of decision and fusion rules, in contrary to DDN-WoF. It can further be assumed that both DDN-WoF and DDN-WF are minimax robust, since this is a desired property, which leads to a little or no loss of performance. Hence, the loss function between DDN-WoF and DDN-WF can be defined as
\begin{align}\label{eq_wf1}
L^K(\cdot)=&R(\pi_0,\boldsymbol{\phi}_r,\gamma^K,\cdot)-R(\pi_0,\boldsymbol{\phi}_r^{*},\gamma_k^K,\cdot)\nonumber\\
=&(C_{01}-C_{00})\left(P_{M_0}(\boldsymbol{\phi}_r,\gamma^K,P_1^\theta)-P_{M_0}(\boldsymbol{\phi}_r^{*},\gamma^K_k,P_1^\theta)\right)\nonumber\\
=&(C_{01}-C_{00})\left(B(K/2;K,1-\theta)-P_{M_0}(\boldsymbol{\phi}_r^{*},\gamma^K_k,P_1^\theta)\right),
\end{align}
where $\gamma^K$ is the majority voting rule and $\gamma_k^K$ is the K-out-of-N fusion rule with overall error probabilities as given in Section~\ref{multisensordefinition2}. Notice that $\boldsymbol{\phi}_r$ and $\boldsymbol{\phi}_r^{*}$ are not the same decision rules, i.e. each component of $\boldsymbol{\phi}_r$ leads to $\theta=P_M=P_F$ but $\boldsymbol{\phi}_r^{*}$ results in $P_F=P_{F_i}$ and $P_M=P_{M_i}$ (not necessarily $P_M=P_F$) for every decision maker $i$. Hence, the optimization for DDN-WF runs over all possible fusion thresholds $k$ and pairs $(P_F,P_M)$ jointly. Now, the aim is to provide a solution to the following problem:
\begin{prb}\label{eq_wf2}
What is the maximum loss of performance between minimax DDN-WoF and minimax DDN-WF, i.e. the maximum of $L(\cdot)$ over all $(P_0^\theta,P_1^\theta)\in ({\mathcal P}_0\times{\mathcal P}_1)(\theta)$, $\boldsymbol{\phi}_r^{*}$ and $\gamma^K_k$, in terms of $\theta$ and in general?
\end{prb}
Determining a solution to Problem~\ref{eq_wf2} is equivalent to finding the maximum gain achievable if DDN-WoF is re-designed to be in the form of a DDN-WF having no information about the observation statistics.


\subsection{Derivation of the Maximum Performance Loss}\label{derivationofthemaximumperformanceloss00741}
It can be seen that the maximization of $L$ in \eqref{eq_wf1} is equivalent to the minimization of $P_{M_0}(\boldsymbol{\phi}_r^{*},\gamma^K_k,P_1^\theta)$ over all $\boldsymbol{\phi}_r^{*}$, $\gamma^K_k$ and $(P_0^\theta,P_1^\theta)\in({\mathcal P}_0\times{\mathcal P}_1)(\theta)$. Since the decision makers are identical, for every $(P_0^\theta,P_1^\theta)\in({\mathcal P}_0\times{\mathcal P}_1)(\theta)$, the decision rule $\boldsymbol{\phi}_r^{*}$ results in $(P_F,P_M)\in B_\theta$, hence the minimization over the distributions $(P_0^\theta,P_1^\theta)\in({\mathcal P}_0\times{\mathcal P}_1)(\theta)$ and the decision rule $\boldsymbol{\phi}_r^{*}$ can be performed over $(P_F,P_M)\in B_\theta$. It is important to note that the fusion rule $\gamma_k^K$ must also guarantee $P_{M_0}=P_{F_0}$ in addition to the minimization of $P_{M_0}$. From Equations~\eqref{eq11} and \eqref{eq12}, the restriction of $P_{M_0}=P_{F_0}$ yields the set of all possible valid pairs $(P_F,P_M)$ to belong to
\begin{equation}\label{equation00727}
(\mathcal{F}\times \mathcal{M})^K_{k}=\Bigg\{\hspace{-1mm}(P_F,P_M):B(k;K,P_F)+B(k;K,1-P_M)=1\Bigg\}
\end{equation}
where $\mathcal{F}^K_{k}$ is the set of all $P_F$s and $\mathcal{M}^K_{k}$ is the set of all corresponding $P_M$s. Define a function, called minimax fusion function, $h^K_k:[0,1]\rightarrow [0,1]$ with the mapping $\mathcal{F}^K_{k} \overset{h^K_k}{\mapsto} \mathcal{M}^K_{k}$, i.e.
\begin{equation}\label{equation00727a}
P_M=h^K_k(P_F)=1-B^{-1}(1-B(P_F)).
\end{equation}
Now only the pair $(B_\theta,h^K_k)$ can further be considered for minimization, because the minimization over $(P_0^\theta,P_1^\theta)\in({\mathcal P}_0\times{\mathcal P}_1)(\theta)$ and $\boldsymbol{\phi}_r^{*}$ is restricted to $(P_F,P_M)\in B_\theta$ and the minimization over $\gamma_k^K$ is restricted to $h^K_k$.\\
The properties of $B_\theta$ has been introduced in Section~\ref{sec4}. In the following, the aim is to evaluate the properties of the minimax fusion function $h^K_k$ and determine which points in $\mathcal{B}_\theta$ satisfy $h^K_k$ and from those points to find the point that minimizes $P_{M_0}$ (or equivalently $P_{F_0}$). It can be seen that the function $h^K_k$ is continuous, passes through the points $(P_F,P_M)\in\{(0,0),(1,1)\}$ for every $K$ and $k< K$. Therefore, at least one point in $(P_F^*,P_M^*)\in\mathcal B_\theta$ satisfies $P_M^*=h^K_k(P_F^*)$. An example of $h^K_k$ is illustrated in Figure~\ref{fig3} for $K=5$ and $k\in\{0,1,2\}$; the cases of $k\in\{3,4\}$ are omitted for the sake of clarity, since $h^K_k$ is symmetric with respect to $P_M=P_F$ (this will be proven later). The following remark and Proposition~\ref{proposition00741} will be used to prove the monotonicity properties of $h^K_k$ in Proposition~\ref{proposition00742}.
\begin{rem}\label{kor1}
$P_{F_0}(P_F;K,k)$ is a monotonically decreasing and $P_{M_0}(P_M;K,k)$ is a monotonically increasing function of $k$.
\end{rem}
\begin{prop}\label{proposition00741}
$P_{F_0}(P_F;K,k)$ and $P_{M_0}(P_M;K,k)$ are monotonically increasing functions of $P_F$ and $P_M$ respectively.
\end{prop}
A proof of Proposition~\ref{proposition00741} can be found in Appendix~\ref{proofofproposition00741}.
\begin{prop}\label{proposition00742}
For every $K$ and $k$, $h^K_k$ is a monotonically increasing function and in particular if $k\in\{0,\ldots,\lfloor{K}/{2}\rfloor-1\}$, $h^K_k(P_F)>P_F$ and if $k\in\{\lfloor {K}/{2}\rfloor+1,\ldots,K-1\}$, $h^K_k(P_F)<P_F$ for all $P_F\in(0,1)$.
\end{prop}
A proof of Proposition~\ref{proposition00742} can be found in Appendix~\ref{proofofproposition00742}. In addition to the monotonicity properties of $h^K_k$, it is also a symmetric function; see the following remark:
\begin{rem}\label{remark00742}
The function $h^K_k$ is symmetric with respect to $k=\lfloor{K}/{2}\rfloor$, meaning that it accepts an inverse function. That is for every $m\in\{0,\ldots,\lfloor{K}/{2}\rfloor\}$, if $P_M=h^K_{\lfloor{K}/{2}\rfloor-m}(P_F)$ then, ${P_F=h^K_{\lfloor{K}/{2}\rfloor+m}(P_M)}$. The assertion follows from the condition that creates $(\mathcal{F}\times \mathcal{M})^K_{k}$. For $m=0$, we have $P_{M_0}(x;K,k)=P_{F_0}(x;K,k)$. Thus, for every $(P_F,P_M)\in(\mathcal{F}\times \mathcal{M})^K_{k}$, it is true that ${P_F=P_M}$, i.e. $h^K_{\lfloor{K}/{2}\rfloor}$ and its inverse are the same. Similarly, when $m>0$,
\begin{align*}
P_{M_0}\left(x;K,\lfloor{K}/{2}\rfloor+m\right)=P_{F_0}\left(x;K,\lfloor{K}/{2}\rfloor-m\right),\\ P_{M_0}\left(x;K,\lfloor{K}/{2}\rfloor-m\right)=P_{F_0}\left(x;K,\lfloor{K}/{2}\rfloor+m\right)
\end{align*}
prove the symmetry of $h^K_k$ on $[0,1]^2$.
\end{rem}
To date all the necessary properties of the set ${B}_\theta$ and the minimax fusion function $h^K_k$ have been derived. It is now possible to state the main theorem of this section.
\begin{thm}\label{theorem00743}
For a fixed $\theta$, among all $(P_0^\theta,P_1^\theta)\in ({\mathcal P}_0\times{\mathcal P}_1)(\theta)$, and all minimax fusion rules $h^K_k$, minimum error probability, i.e. the maximum of \eqref{eq_wf1}, is achieved by the fusion threshold $k=0$, $($or $k=K-1$$)$ and a point on $l_1^*$ $($or on $l_2^*)$. Hence, the maximum of performance loss due to the absence of a fusion center for $C_{01}-C_{00}=1$ is given by
\begin{equation}\label{equation00729}
L^K(\theta)=\frac{1}{1+\hat{\theta}^K}-B\left(\lfloor{K}/{2}\rfloor;K,\theta\right).
\end{equation}
Moreover,
\begin{equation}\label{equation00730}
\lim_{K\rightarrow \infty}\arg \sup_{\theta\in(0,1/2)} \; L^K(\theta)=\lim_{K\rightarrow \infty}\sup_{\theta\in(0,1/2)}\; L^K(\theta)=\frac{1}{2}.
\end{equation}
\end{thm}

\begin{proof}
As mentioned before, the maximum of \eqref{eq_wf1} requires joint minimization over $\mathcal{B}_\theta$ and $h_k^K$. Due to the properties of $h_k^K$ stated by Propositions~\ref{proposition00741} and \ref{proposition00742}, the minimization over $\mathcal{B}_\theta$ can further be confined to the line segments $l_1^*$ and $l_2^*$. The details are as follows: continuity and monotonicity of $h_k^K$ on $[0,1]^2$ guarantees that for every $K$, $h_k^K$ intersects only a single point $({P_F}_k^K,{P_M}_k^K)$, which is either on $l_1^*$ or on $l_2^*$. Examples of $({P_F}_k^K,{P_M}_k^K)$ can be seen in Figure~\ref{fig3}, e.g. $({P_F}_1^5,{P_M}_1^5)\approx (0.1,0.4)$, which is the intersection point of $h_1^5$ with $l_1$. Again by the monotonicity of $h_k^K$, all other points belonging to $\mathcal{B}_\theta$ and intersected by $h_k^K$ have higher $P_F$ and $P_M$. From Proposition~\ref{proposition00741}, $P_{F_0}$ and $P_{M_0}$ are increasing in $P_F$ and $P_M$, respectively. Therefore, the minimization over $\mathcal{B}_\theta$ reduces to a minimization over $l_1^*$ and $l_2^*$. The minimization can further be reduced to either $l_1^*$ and $h_k^K$ for $k\in \{0,\ldots,\lfloor{K}/{2}\rfloor\}$, or $l_2^*$ and $h_k^K$ for $k\in \{\lfloor{K}/{2}\rfloor,\ldots,K-1\}$. This result follows from the symmetry of $\mathcal{B}_\theta$, and $h_k^K$ with respect to $P_M=P_F$, cf. Remark~\ref{remark00742}, Figure~\ref{fig3}. Since both choices result in the same error probability, let us consider $l_1^*$ and $h_k^K$ for $k\in \{0,\ldots,\lfloor{K}/{2}\rfloor\}$, and generalize the results thereafter. By doing so, the intersection point of $h_k^K$ with the line segment $l_1^*$ is $({P_F}_k^K,\hat{\theta}(1-{P_F}_k^K))$. Since $h^K_k$ passes through this point, this point must satisfy Equation~\eqref{equation00727a}. Let $x_k^K=1-{P_F}_k^K$, then Equation~\eqref{equation00727a} for $(1-x_k^K,\hat{\theta}x_k^K)$ can be written as
\begin{equation}\label{equation00731}
x_k^K=\left(\frac{1}{\hat{\theta}^K f_1(x_k^K,\cdot)+f_2(x_k^K,\cdot)}\right)^{1/K},
\end{equation}
where
\begin{align*}
f_1(x_k^K,\cdot)&=(\hat{\theta}x_k^K)^{-K}B(k;K,1-\hat{\theta}x_k^K),\\
f_2(x_k^K,\cdot)&=(x_k^K)^{-K}B(k;K,1-x_k^K).
\end{align*}
The difference in error probability for the cases $k=0$ and $k>0$ can then be determined by
\begin{align}\label{equation00734}
L^K(\cdot)=&P_{M_0}(\hat{\theta}x_k^K;K,k)-P_{M_0}(\hat{\theta}x_0^K;K,0)\nonumber\\
=&(\hat{\theta}x_k^K)^{K} f_1(x_k^K,\cdot)-\hat{\theta}^K/(1+\hat{\theta}^K)\nonumber\\
=&\frac{\hat{\theta}^K}{1+\hat{\theta}^K}\left(\frac{\hat{\theta}^K f_1(x_k^K,\cdot)+f_1(x_k^K,\cdot)}{\hat{\theta}^K f_1(x_k^K,\cdot)+f_2(x_k^K,\cdot)}-1\right)>0,
\end{align}
where the second equality follows from \eqref{eq12} and \eqref{equation00731}, the third equality follows from \eqref{equation00731} and the last inequality follows from $f_1(x_k^K,\cdot)-f_2(x_k^K,\cdot)>0$, since $0<\hat{\theta}<1$. This proves that $h_0^K$ and the point $(1-1/(1+\hat\theta^K)^{1/K},\hat\theta/(1+\hat\theta^K)^{1/K})$ on $l_1^*$ minimizes the error probability. Due to the symmetry of the problem, this result is equivalent to $h_{K-1}^K$ and $(\hat\theta/(1+\hat\theta^K)^{1/K},1-1/(1+\hat\theta^K)^{1/K})$ on $l_2^*$. Equation~\ref{equation00729} is then immediate (and is a very special case) from \eqref{equation00734} with $k=\lfloor{K}/{2}\rfloor$ and noting that for this choice, $\theta={P_M}_{\lfloor{K}/{2}\rfloor}^K={P_F}_{\lfloor{K}/{2}\rfloor}^K$. A proof of \eqref{equation00730} is given in Appendix~\ref{proofofequation00730}.
\end{proof}
The derivations indicate that given $\theta$, the lack of fusion center amounts to a maximum loss of performance given by \eqref{equation00729}, without having any knowledge about the observation statistics. In case $\theta$ is also unknown, the maximum loss $L^K$ tends to $1/2$, when $K$ and $\theta$ are large enough. This can also be seen in Figure~\ref{figure0076}, where $L^K$ is illustrated for various $K$. Therefore, it is theoretically possible that while DDN-WoF gives decisions that are equivalent to tossing a coin, DDN-WF can give decisions that are free of errors. The fusion rules which provide this property are AND ($k=0$) and OR ($k=\lfloor{K}/{2}\rfloor-1$) fusion rules, which are widely used in many practical applications due to their simplicity.
\begin{figure}[ttt]
  \centering
  \centerline{\includegraphics[width=8.8cm]{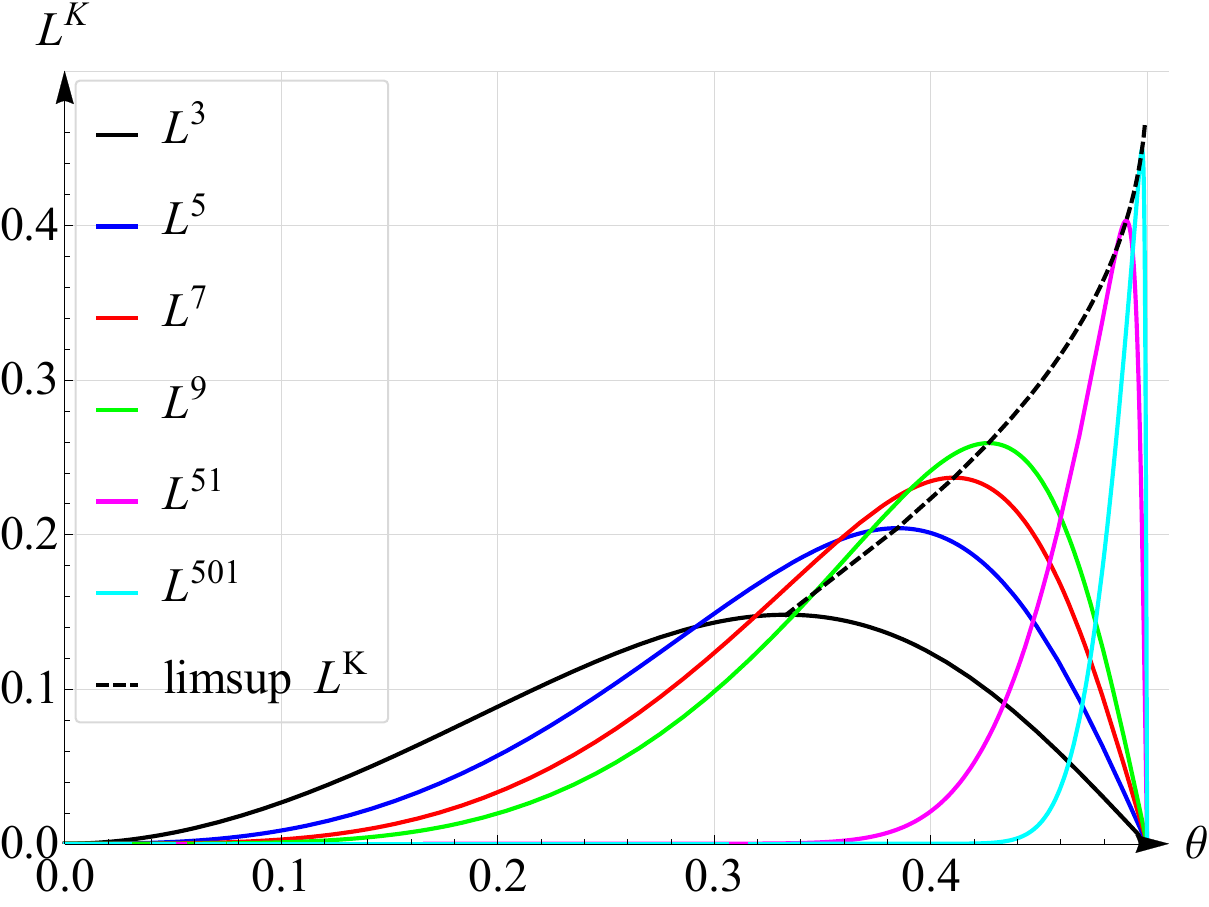}}
\caption{Maximum performance loss between minimax DDN-WoF and minimax DDN-WF for various $K$.\label{figure0076}}
\end{figure}
\subsection{Generalizations}\label{generalizations00742}
There are two possible generalizations, which are detailed as follows:
\subsubsection{Non-identical Decisions Scenario}\label{nonidenticaldecisionsscenario007421}
The results obtained above can be generalized to non-identical decisions scenario. Consider $K$ to be sufficiently large. It was shown that $k=0$ or $k=\lfloor{K}/{2}\rfloor-1$ maximizes $L^K$. Fixing $k=0$ and allowing non-identical decisions, the condition $P_{F_0}=P_{M_0}$ reduces to
\begin{equation}\label{equation00735}
\prod_{i=1}^K P_{M_i}+\prod_{i=1}^K (1-P_{F_i})=1.
\end{equation}
When $K$ is large, either $P_{F_i}$ must be small or $P_{M_i}$ must be large such that \eqref{equation00735} holds. The
$P_{M_i}$ are large only on $l_2$, for which the $P_{F_i}$ are also small. This eventually indicates that either \eqref{equation00735} does not hold or if it holds, the error probability is higher than that of the one obtained by considering the points on $l_1$. Hence, for each sensor, $(P_{F_i},P_{M_i})$ can be sampled from $l_1$,
which is defined to satisfy $1-P_{F_i}=P_{M_i}/\theta$. Inserting this into \eqref{equation00735} gives
\begin{equation*}
\prod_{i=1}^K P_{M_i}=\frac{\hat{\theta}^K}{1+\hat{\theta}^K}=P_{M_0}(\hat{\theta}x_0^K;K,0),
\end{equation*}
which is the same error probability with the one that is obtained by the identical local decision makers constraint.
\subsubsection{Comparison Regarding the Number of Sensors}\label{comparisonregardingthenumberofsensors007422}
There is a possible generalization of the loss of performance in terms of the number of sensors.
If advantageous, an observer may prefer to increase the number of sensors instead of re-designing the network with a fusion center. Let $K_1$ be the number of sensors for DDN-WF and $K_2$ is the number of sensors for DDN-WoF. Then, for the worst case analysis,
\begin{equation}\label{eqgen2}
P_{M_0}(\theta;K_2,K_2/2)={\hat{\theta}^{K_1}}/({1+\hat{\theta}^{K_1}}).
\end{equation}
Solving Equation~\eqref{eqgen2} for $K_1$ gives
\begin{equation*}
K_1=\left\lceil\frac{\log\hat{P}_{M_0}}{\log\hat\theta}\right\rceil,\quad \hat{P}_{M_0}=\frac{P_{M_0}(\theta;K_2,K_2/2)}{1-P_{M_0}(\theta;K_2,K_2/2)}.
\end{equation*}
The relation between $K_1$ and $K_2$ is exponential, e.g. for $\theta\approx0.19$, if $K_1=19$, then $K_2=101$. Therefore, the change of the network is more practical, especially if DDN-WoF possesses a large number of sensors.

\section{Conclusions}\label{conclusions}
Implications of minimax detection has been studied for two cases of decentralized detection networks, with and without a fusion center. Under a Bayesian setting of the hypothesis testing problem, DDN-WoF is composed of a majority fusion rule and identical local decision makers, each resulting in an error probability of $\theta=P_F=P_M$. For this setup, it has been shown that increasing the total number of sensors from an odd number to an even number does not lead to an increase in detection performance. This result is counter-intuitive, because large deviations theory indicates that the error probability decreases exponentially with the increase of the number of observations, which is analogous to the number decision makers in DDN-WoF. Another implication of minimax decision making is on the loss of detection performance in terms of average error probability. The bounds derived for a single decision maker and for multiple decision makers indicate no significant loss of detection performance due to minimax decision making. Another concern is that single sensor optimum decision rules rapidly become non optimal for multi-sensor systems. In many applications DDN-WoF is considered as a default distributed detection scheme due to its physical robustness properties and the exponential decay rate of error probability provided by the large deviations theory. The results indicate that when the number of decision makers $K$, and the average error probability $\theta<1/2$ are large enough, it is possible that DDN-WoF outputs random decisions while DDN-WF is error free. This result has been generalized to the case, where the decision makers are not necessarily identical. Another generalization suggests that for the worst case analysis, the same detection performance between DDN-WoF and DDN-WF can be obtained if the number of sensors for DDN-WoF is exponentially larger than that of DDN-WF. It is then more appealing to re-design DDN-WoF in the form of DDN-WF, instead of increasing the number of sensors. Our findings are theoretical and we believe that they are crucial before, during, and after designing minimax robust decentralized detection systems.

\appendices\label{appx}
\section{Proof of Proposition\ref{prop1}}\label{appendix1}
\begin{proof}\label{proof1}
Given $P_F=P_{F_i}$ and $P_M=P_{M_i}$ for all $i\in\{1,\ldots,K\}$, false alarm and miss detection probabilities resulting from $\gamma^K$ are $P_{F_0}=1-B(K/2;K,P_F)$ and $P_{M_0}=B(K/2;K,1-P_M)$, respectively, where $B(k;K,P)$ is a binomial cumulative distribution function with at most $k$ successes out of $K$ trials each having a success probability $P$. Let $X\sim B(K,P)$ and ${Y\sim B(K,1-P)}$ be two Binomial r.v.s with $K$ trials each having a success probability $P$ and $1-P$, respectively. Then, for two disjoint events $E_1=X\leq K/2$ and $E_2=(K-X)\leq K/2$,
\begin{equation*}
P(E_1 \cup E_2)=P(E_1)+P(E_2)=1.
\end{equation*}
Now, by noting that $Y=K-X$ in distribution, we have
\begin{align*}
P(E_1)+P(E_2)&=P(X\leq K/2)+P(Y\leq K/2)\nonumber\\
&=B(K/2;K,P)+B(K/2;K,1-P)=1
\end{align*}
which implies that $P_{F_0}$ and $P_{M_0}$ own the same polynomial function $f$ s.t. $P_{F_0}=f(P_F)$ and $P_{M_0}=f(P_M)$. From Proposition\ref{prop1}, $f$ is monotonically increasing, hence $P_{F_0}=P_{M_0}$ iff $P_F=P_M$.
\end{proof}

\section{Proof of Proposition\ref{prop2}}\label{appendix2}
\begin{proof}\label{proof2}
Using the substitution $j=i+1$, we have
\begin{align}\label{equationa3}
P_F P_{F_0}^{2K-1} &= \sum_{i=K}^{2K-1} \binom{2K-1}{i} {P_F}^{i+1}(1-P_F)^{2K-1-i}\nonumber\\
&= \sum_{j=K+1}^{2K} \binom{2K-1}{j-1} {P_F}^{j}(1-P_F)^{2K-j}\nonumber\\
&= \left(\sum_{j=K+1}^{2K-1}\binom{2K-1}{j-1} {P_F}^{j}(1-P_F)^{2K-j} \right)+ {P_F}^{2K}
\end{align}
and
\begin{align}\label{equationa4}
(1-P_F)P_{F_0}^{2K-1}&= \sum_{i=K}^{2K-1}\binom{2K-1}{i}{P_F}^i(1-P_F)^{2K-i}\nonumber\\
&=\binom{2K-1}{K}{P_F}^K(1-P_F)^K \nonumber\\
&+\sum_{i=K+1}^{2K-1} \binom {2K-1}{i}{P_F}^i(1-P_F)^{2K-i}.
\end{align}
Adding up \eqref{equationa3} and \eqref{equationa4}, we get
\begin{align*}
&(1-P_F)P_{F_0}^{2K-1} + P_FP_{F_0}^{2K-1}\nonumber\\
&= \binom{2K-1}{K}P_F^K(1-P_F)^K\nonumber \\
&+ \sum_{i=K+1}^{2K-1}\left(\binom{2K-1}{i-1}+\binom{2K-1}{i}\right){P_F}^{i}(1-P_F)^{2K-i} +{P_F}^{2K}\nonumber\\
&=\frac{1}{2}\binom {2K} K {P_F}^K(1-P_F)^K+\sum_{i=K+1}^{2K}\binom {2K} i {P_F}^i(1-P_F)^{2K-i}\nonumber\\
&=P_{F_0}^{2K}
\end{align*}
using the identities
\begin{equation*}
\binom{2K-1}{K}=\frac{1}{2}\binom{2K}{K},\quad\binom{2K-1}{i}+\binom{2K-1}{i-1}=\binom{2K}{i}
\end{equation*}
and
\begin{equation*}
{P_F}^{2K}=\binom{2K}{2K} {P_F}^{2K}(1-P_F)^{2K-2K}.
\end{equation*}
\end{proof}

\section{Proof of Lemma~\ref{lemma1}}\label{appendix3}
\begin{proof}\label{proofoflemma00732}
By definition, $P_F$ and $P_M$ are probabilities, hence $(P_F,P_M)\in[0,1]^2$. Evaluating $P_F=1-P_0[l(Y)\leq t]$ and $P_M=P_1[l(Y)\leq t]$ for $\lim_{t\to\ 0}$ and $\lim_{t\to\ \infty}$ shows that $r_t$ passes through the points $(1,0)$ and $(0,1)$. Let $p_{0,l}$ and $p_{1,l}$ be the density functions of $l(Y)$ for $Y\sim P_0$ and $Y\sim P_1$, respectively. Since $r_t$ is differentiable for every $t$, i.e.
\begin{equation}\label{equationa7}
\frac{d P_M}{d P_F}=\frac{d P_M}{d t}\frac{d t}{d P_F}=-\frac{p_{1,l}(t)}{p_{0,l}(t)}
\end{equation}
exists, $r_t$ is continuous. The miss detection probability can also be written as
\begin{align*}
P_M&=\int_{\{y:l(y)\leq t\}}p_1(y)\mathrm{d}y=\int_{\{y:l(y)\leq t\}}l(y)p_0(y)\mathrm{d}y\nonumber\\
&=\int_0^t x p_{0,l}(x)\mathrm{d}x,
\end{align*}
where the last equality follows from
\begin{equation*}
p_{0,l}(x)=\left|\frac{\mathrm{d}l^{-1}(x)}{\mathrm{d} x} \right|p_0(l^{-1}(x))
\end{equation*}
with the change of variable $x=l(y)$. Hence,
\begin{equation*}
\frac{d P_M}{d t}=t p_{0,l}(t)\overset{\mbox{\eqref{equationa7}}}{\Longrightarrow} \frac{d P_M}{d P_F}=-t.
\end{equation*}
As a result,
\begin{equation*}
\frac{d^2 P_M}{d {P_F}^2}=\frac{d}{d {P_F}}\left(\frac{d P_M}{d {P_F}}\right)=-\frac{d t}{d {P_F}}=\frac{1}{p_{0,l}(t)}\geq 0
\end{equation*}
proves that $r_t$ is convex.
\end{proof}

\section{Proof of Proposition\ref{prop3}}\label{appendix4}
Letting $p=p(\theta)=1-\theta$ and $q=q(\theta)={1-2\theta\over 1-\theta}$, showing that $L^K_\infty$ is negative for sufficiently large $K$ is equivalent to showing that
\begin{equation}\label{eq_app4_1}
\sum_{i=0}^{K/2} {K\choose i} p^i (1-p)^{K-i}<{1\over 2}\sum_{i=0}^{K/2} {K\choose i} q^i (1-q)^{K-i}
\end{equation}
for sufficiently large $K$.
There are two possible cases:
\begin{itemize}
\item Trivial case: For $1/3\leq \theta <1/2$, the sum on the left converges to zero and the sum on the right converges
to a positive number, so the inequality \eqref{eq_app4_1} is true for large $K$.
\item Remaining case: Suppose $0< \theta <1/3$. The inequality
of the sums can be proven working term by term. It suffices to show that
\begin{equation}\label{eq_app4_2}
p^i(1-p)^{K-i}<{1\over 2}q^i(1-q)^{K-i}
\end{equation}
for all $0\leq i\leq K/2$, when $K$ is large enough. Note that ${p(1-q)\over q(1-p)}={1-\theta\over 1-2\theta}> 1$ and
${p(1-p)\over q(1-q)}={(1-\theta)^3\over 1-2\theta}< 1$.
Therefore,
\begin{align}\label{eq_app4_3}
\left({1-p\over 1-q}\right)^K\left(p(1-q)\over q(1-p)\right)^i
&\leq \left({1-p\over 1-q}\right)^K\left(p(1-q)\over q(1-p)\right)^{K/2}\nonumber\\
&= \left({p(1-p)\over q(1-q)}\right)^{K/2}.
\end{align}
\end{itemize}
The right hand side of \eqref{eq_app4_3} can be made less than $1/2$ by taking $K$ sufficiently large, giving the inequality \eqref{eq_app4_2} and hence the inequality \eqref{eq_app4_1}.

\section{Proof of Proposition~\ref{proposition00741}}
\begin{proof}\label{proofofproposition00741}
To prove that $P_{F_0}(P_F;K,k)$ and $P_{M_0}(P_M;K,k)$ are increasing functions of $P_F$ and $P_M$, respectively, it is sufficient to prove it only for $P_{F_0}(P_F;K,k)$. Because
\begin{equation}\label{equationa11}
\frac{\partial{P_{M_0}\left(P_M;K,k\right)}}{\partial{P_{M}}}=\frac{\partial{P_{F_0}\left(P_F;K,k\right)}}{\partial{P_{F}}}|_{P_F:=1-P_M}.
\end{equation}
Noting that \eqref{equationa11} is zero for $P_F=0$, we have
\begin{align}\label{equationa12}
P_{F_0}&=\sum^K_{i=k} \binom{K}{i}P_F^i (1-P_F)^{K-i}\nonumber\\
&=\sum^K_{i=k} \binom{K-1}{i-1}P_F^i (1-P_F)^{K-i}\nonumber\\&+\sum^K_{i=k} \binom{K-1}{i}P_F^i (1-P_F)^{K-i}.
\end{align}
Since in the second sum, the term is zero when $i=K$, we get
\begin{align}\label{equationa13}
&\sum^K_{i=k} \binom{K-1}{i}P_F^i (1-P_F)^{K-i}\nonumber\\
&=\sum^{K-1}_{i=k} \binom{K-1}{i}P_F^i (1-P_F)^{K-i}\nonumber\\&<\sum^{K-1}_{i=k-1} \binom{K-1}{i}P_F^i (1-P_F)^{K-i}.
\end{align}
Changing the variable $j=i+1$,
 \begin{align}\label{equationa14}
&\sum^{K-1}_{i=k-1} \binom{K-1}{i}P_F^i (1-P_F)^{K-i}\nonumber\\
&=\sum^{K}_{j=k} \binom{K-1}{j-1}P_F^{j-1} (1-P_F)^{K-j+1}
\end{align}
and writing \eqref{equationa14} in \eqref{equationa12} with \eqref{equationa13}, it follows that
\begin{align}\label{equationa15}
\sum^K_{i=k} \binom{K}{i}{P_F}^i (1-P_F)^{K-i}&<\sum^{K}_{i=k} \binom{K-1}{i-1}{P_F}^i (1-P_F)^{K-i}\nonumber\\&+\sum^{K}_{j=k} \binom{K-1}{j-1}{P_F}^{j-1} (1-P_F)^{K-j+1}.
\end{align}
Using
\begin{equation*}
{P_F}^i(1-P_F)^{K-i}+{P_F}^{i-1}(1-P_F)^{K-i+1}={P_F}^{i-1}(1-P_F)^{K-i}
\end{equation*}
rewrite \eqref{equationa15},
\begin{equation}\label{equationa17}
\sum^K_{i=k} \binom{K}{i}{P_F}^i (1-P_F)^{K-i}<\sum^{K}_{i=k} \binom{K-1}{i-1}{P_F}^{i-1} (1-P_F)^{K-i}.
\end{equation}
Multiplying \eqref{equationa17} with $K/(1-P_F)$ and noting that
\begin{equation*}
i\binom{K}{i}=K\binom{K-1}{i-1}
\end{equation*}
we finally get
\begin{equation*}
\sum^K_{i=k} \binom{K}{i}{P_F}^{i-1} (1-P_F)^{K-i-1}\left(i-KP_F\right)=\frac{\partial{P_{F_0}\left(P_F;K,k\right)}}{\partial{P_{F}}}>0.
\end{equation*}
\end{proof}

\section{Proof of Proposition~\ref{proposition00742}}
\begin{proof}\label{proofofproposition00742}
The claim will be proven for odd $k$, while its extension to even $k$ can be accomplished following the same line of arguments.
Let the threshold be $k\in\{0,\lfloor {K}/{2}\rfloor -1\}$ for some $K$. If $k=\lfloor{K}/{2}\rfloor$, then clearly
\begin{equation*}
P_{F_0}\left(x;K,\lfloor{K}/{2}\rfloor\right)=P_{M_0}\left(x;K,\lfloor{K}/{2}\rfloor\right),\quad\forall x\in[0,1].
\end{equation*}
One can also see that, cf.~Remark~\ref{kor1},
\begin{equation*}
P_{F_0}\left(x;K,\lfloor{K}/{2}\rfloor-1\right)>P_{F_0}\left(x;K,\lfloor{K}/{2}\rfloor\right),\quad\forall x\in(0,1),
\end{equation*}
and
\begin{equation*}
P_{M_0}\left(x;K,\lfloor{K}/{2}\rfloor-1\right)<P_{M_0}\left(x;K,\lfloor{K}/{2}\rfloor\right),\quad\forall x\in(0,1).
\end{equation*}
Hence,
\begin{equation}\label{equationa23}
P_{F_0}\left(x;K,\lfloor{K}/{2}\rfloor-1\right)>P_{M_0}\left(x;K,\lfloor{K}/{2}\rfloor-1\right),\,\,\forall x\in(0,1).
\end{equation}
For a pair $(P_F,P_M)$ to be valid, it should be in $(\mathcal F\times \mathcal M)^K_{k}$, i.e.
\begin{equation}\label{equationa24}
P_{F_0}\left(P_F;K,\lfloor{K}/{2}\rfloor-1\right)=P_{M_0}\left(P_M;K,\lfloor{K}/{2}\rfloor-1\right).
\end{equation}
Assume that \eqref{equationa24} holds for some $(P_F^*,P_M^*)$ with $P_M^*=P_F^*$ or with $P_M^*<P_F^*$. Then, both
cases are obviously a contradiction with \eqref{equationa23}, since both $P_{F_0}$ and $P_{M_0}$ are monotonically increasing functions of $P_F$ and $P_M$, respectively, cf. Proposition~\ref{proposition00741}. Therefore, $P_M^*>P_F^*$ must be true for all pairs $(P_F^*,P_M^*)\in (\mathcal F\times \mathcal M)^K_{k}$. This proves that $h^K_k(P_F)>P_F$ for all $P_F\in(0,1)$. Clearly, when $k\in\{\lfloor {K}/{2}\rfloor +1,K\}$, due to symmetry, e.g., $P_{M_0}\left(x;K,\lfloor{K}/{2}\rfloor+1\right)=P_{F_0}\left(x;K,\lfloor{K}/{2}\rfloor-1\right)$, the inequalities above change direction and we get $h^K_k(P_F)<P_F$ for all $P_F\in(0,1)$. Next, assume that $(P_F^*,P_M^*)$, $(P_F^*,P_M^*)\neq(1,1)$, is a valid pair that satisfies \eqref{equationa24} and fix a small positive number $\delta$. Since $P_{F_0}$ is increasing,
\begin{align*}
P_{F_0}\left(P_F^*+\delta;K,\lfloor{K}/{2}\rfloor-1\right)>P_{F_0}\left(P_F^*;K,\lfloor{K}/{2}\rfloor-1\right),&\nonumber\\
\forall P_F^*\in[0,1).&
\end{align*}
This suggests that the left hand side of \eqref{equationa24} increases by adding $\delta$ to $P_F^*$. In order \eqref{equationa24} to hold, its right
hand side must also increase, which implies an increase of $P_M^*$ by some positive number
$\epsilon$, since $P_{M_0}$ is also an increasing function. Then, \mbox{$(P_F^*+\delta,P_M^*+\epsilon)\in (\mathcal F\times \mathcal M)^K_k$} for all $(P_F^*,P_M^*)\neq(1,1)$
implies that $h^K_k$ is a monotonically increasing function.\\
\end{proof}

\section{Proof of~\eqref{equation00730}}
\begin{proof}\label{proofofequation00730}
Introducing a random variable $X_K$ with a binomial distribution $B(K,\theta)$, it can be shown that
\begin{equation*}
L^K(\theta)=P[X_K> \lfloor{K}/{2}\rfloor]-\frac{1}{1+{\hat\theta}(\theta)^{-K}}.
\end{equation*}
For every $\theta\leq\frac{1}{2}$, $P[X_K>\lfloor {K}/{2}\rfloor]\leq\frac{1}{2}$ hence $P[X_K> \lfloor{K}/{2}\rfloor]<\frac{1}{2}$.
Assume that $\theta=\theta_K(x)$ where $\theta_	K (x)=\frac{1}{2}\left(1-\frac{x}{\sqrt{K}}\right)$, for some fixed positive $x$.
Then, $\lfloor{K}/{2}\rfloor=E[X_K]+x_K\sigma(X_K)$ with $x_K=x/\sqrt{4\theta_K(x)(1-\theta_K(x))}\sim x$. The central limit theorem implies that
\begin{align*}
 P[X_K> \lfloor {K}/{2}\rfloor]&=P\left[X_K> E[X_K]+x_K\sigma(X_K)\right]\nonumber\\
&=P\left[\frac{X_K-E[X_K]}{\sigma(X_K)}> x_K\right]\nonumber\\
&=P[X^{'}_K> x]=1-F(x)\,\, \mbox{\rm when}\,\,K\rightarrow\infty
\end{align*}
where $X^{'}_K \sim {\cal{N}}(0,\sigma^2)$. Since $\hat\theta(\theta_K(x))^{-K}\rightarrow\infty$ when $K\rightarrow\infty$, we get,
\begin{equation*}
 \lim_{K\rightarrow\infty}\sup_{\theta\leq 1/2}L^K(\theta)\geq\lim_{K\rightarrow\infty}L^K(\theta_K(x))=1-F(x).
\end{equation*}
As $F(x)\rightarrow\frac{1}{2}$ when $x\rightarrow 0^+$, this proves the claim.
\end{proof}

%

\bibliographystyle{IEEEtran}
\bibliography{strings2new}
\end{document}